\documentclass{article}
\usepackage[english]{babel}
\usepackage{indentfirst}
\usepackage{fancyhdr}
\usepackage{algorithmicx}
\usepackage{algorithm,algpseudocode}
\usepackage{float}
\usepackage{amssymb}
\usepackage{amsmath}
\usepackage{amsthm}
\usepackage{amsfonts}
\usepackage{mathrsfs}
\usepackage{graphicx}
\usepackage{subfigure}
\usepackage{caption}
\usepackage[shortlabels]{enumitem}
\usepackage{multicol}
\usepackage[toc]{appendix}
\usepackage{fancyhdr}
\usepackage{midpage}
\usepackage{geometry}
\usepackage{caption}
\usepackage[utf8]{inputenc}
\usepackage[OT1]{fontenc}
\usepackage[T1]{fontenc}
\usepackage{textcomp}
\usepackage{xcolor}
\usepackage{centernot}
\usepackage{colortbl}
\usepackage{caption}
\usepackage{framed}
\usepackage{tabularx}
\usepackage{tikz}
\usepackage{tikz-cd} 
\usepackage{varioref}
\usepackage[hidelinks]{hyperref}
\usepackage[nameinlink, capitalise, noabbrev]{cleveref}

\makeatletter
\def\BState{\State\hskip-\ALG@thistlm}
\makeatother

\newcommand{\comment}[1]{}

\newcommand{\mult}[1]{\mathtt{#1}}

\theoremstyle{plain}

\newtheorem{remark}{Remark}

\newtheorem{thm}{Theorem}[section]
\crefname{thm}{theorem}{theorems}
\Crefname{thm}{Theorem}{Theorems} 

\newtheorem{cor}[thm]{Corollary}
\crefname{cor}{corollary}{corollaries}
\Crefname{lem}{Corollary}{Corollaries}

\crefname{prop}{proposition}{propositions}
\Crefname{prop}{Proposition}{Propositions}

\newtheorem{lem}[thm]{Lemma}
\crefname{lem}{lemma}{lemmas}
\Crefname{lem}{Lemma}{Lemmas}

\theoremstyle{definition} 
\newtheorem{defn}{Definition}[section]
\crefname{defn}{definition}{definitions}
\Crefname{defn}{Definition}{Definitions}
\theoremstyle{remark}
\crefname{remark}{remark}{remarks}
\Crefname{remark}{Remark}{Remarks}

\tikzstyle{mybox} = [draw=black, very thick, rectangle, rounded corners, inner ysep=5pt, inner xsep=5pt]

\newcounter{protocol}
\newenvironment{protocol}[1]
{\par\addvspace{\topsep}
	\noindent
	\tabularx{\linewidth}{@{} X @{}}
	\hline
	\refstepcounter{protocol}\textbf{Protocol \theprotocol} #1 \\
	\hline}
{ \\
	\hline
	\endtabularx
	\par\addvspace{\topsep}}
\Crefname{protocol}{Protocol}{Protocols}

\title{Multidimensional Byzantine Agreement in a Synchronous Setting}
\author{Andrea Flamini \and Riccardo Longo \and Alessio Meneghetti}
\date{Department of Mathematics, University of Trento\\ via Sommarive, 14 - 38123 Povo (Trento), Italy \thanks{andrea.flamini.1995@gmail.com, riccardolongmath@gmail.com, alessio.meneghetti@unitn.it}}

\begin{document}

\maketitle

\abstract{In this paper we present the \emph{Multidimensional Byzantine Agreement (MBA) Protocol}, a \emph{leaderless} Byzantine agreement protocol defined for complete and synchronous networks that allows a network of nodes to reach consensus on a vector of relevant information regarding a set of observed events.

The consensus process is carried out in parallel on each component, and the output is a vector whose components are either values with wide agreement in the network (even if no individual node agrees on every value) or a special value $\bot$ that signals irreconcilable disagreement.
The MBA Protocol is probabilistic and its execution halts with probability 1, and the number of steps necessary to halt follows a Bernoulli-like distribution.

The design combines a \emph{Multidimensional Graded Consensus} and a \emph{Multidimensional Binary Byzantine Agreement}, the generalization to the multidimensional case of two protocols presented by Micali et al. in \cite{feldman1997optimal,micali2016byzantine}.

We prove the correctness and security of the protocol assuming a synchronous network where less than a third of the nodes are malicious.
}

\section{Introduction}
The notion of \emph{Byzantine agreement} was introduced for the
binary case (i.e. when the initial value consists of a bit) by Lamport, Shostak, and Pease \cite{lamport2019byzantine}, then quickly extended
to arbitrary initial values (see the survey of Fischer \cite{fischer1983consensus}).
A (binary) \emph{Byzantine agreement protocol} or \emph{Byzantine Fault Tolerant (BFT) protocol}, is a protocol that allows a set of mutually mistrusting players to reach agreement on an arbitrary (respectively binary) value.
These protocols have been initially developed to deal with \emph{Byzantine faults} in distributed computing systems.
A Byzantine fault is a particularly tricky failure where a component, such as a server, can inconsistently appear both failed and functioning, presenting different symptoms to different observers.
The problem then evolved to model malicious behaviour in distributed and multi-party protocols, with natural applications in distributed ledgers such as blockchains, alongside Proof of Work and Proof of Stake solutions \cite{longo2020analysis,meneghetti2020survey,PoS}.

Many BFT protocols \cite{buterin2017casper,castro1999practical,team2017zilliqa,yin2018hotstuff} use the \emph{primary-backup} model, pioneered in the Practical Byzantine Fault Tolerance (PBFT) protocol proposed by Castro and Liskov \cite{castro1999practical}.
These BFT protocols are designed in a way that a single replica is designated as the primary and is responsible for coordinating the consensus decisions, while all the other replicas perform the backup role. However, this primary can be smartly malicious and degrade the performance of the system. For this reason there have been various efforts to design leaderless BFT protocols \cite{micali2016byzantine,rocket2018snowflake}.

\paragraph{Outline}
In \Cref{Net_Ass} we formally define a Byzantine agreement protocol and introduce our assumptions on the network.
Then, in \Cref{probDesc} we give a motivation for our generalization work by giving a model that describes interesting real-world applications and giving an example where existing solutions fall short to expectations.

In \Cref{prot-comp} we proceed to define our protocol components, starting with some useful notation.
In \Cref{MBBA subsection} we extend to the multidimensional case the Binary Byzantine Agreement protocol \cite{micali2016byzantine} defining the \emph{Multidimensional Binary Byzantine Agreement (MBBA) Protocol} and we prove it satisfies the properties of a Byzantine agreement protocol.
In \Cref{MGC subsection} we  extend the notion of $(n,t)$-Graded Consensus protocol introduced by Micali in \cite{chen2016algorand}, then we define the \emph{Multidimensional Graded Consensus (MGC) Protocol}, the natural extension of the Graded Consensus Protocol adopted in Algorand \cite{chen2016algorand}, whose definition comes from the Gradecast Protocol presented by Micali in \cite{feldman1997optimal}.

Then, in \Cref{MBA section} we combine the MGC Protocol and the MBBA Protocol into the \emph{Multidimensional Byzantine Agreement (MBA) Protocol}, the extension to the multidimensional case of the Byzantine agreement Protocol in a synchronous setting described in Algorand \cite{chen2016algorand},
proving that it satisfies the properties of Byzantine agreement protocol.
In \Cref{PG} we introduce a probability game that models how the consensus is reached during the protocol execution, alongside a probability distribution.
Then, in \Cref{mainT} we prove that this distribution predicts the number of steps necessary to end the protocol execution.

Finally, in \Cref{conclusion} we draw some conclusions and remarks, and outline future works to improve the applicability of the MBA protocol in practical settings.

\subsection{BA Definition and Network Assumptions}
\label{Net_Ass}
We now provide the formal definition of Byzantine Agreement protocol:
\begin{defn}[\textit{$(n,t)$-Byzantine Agreement protocol}]
We say that $\mathcal{P}$ is an arbitrary-value (respectively, binary) \emph{$(n, t)$-Byzantine Agreement} (BA) protocol with soundness ${0<\sigma<1}$ if, for every set of values $V$ not containing the special symbol $\bot$ (respectively, for $V = \{0, 1\}$), in an
execution in which at most $t$ out of the $n$ players are malicious and every player $i$ starts with an
initial value $v_i \in V$, every honest player $j$ halts with probability 1, outputting a value $o_j \in V\cup\{\bot\}$
so as to satisfy, with probability at least $\sigma$, the following two conditions:
\begin{enumerate}
    \item Agreement: there exists $o \in V\cup\{\bot\}$ such that $o_j = o$ for each honest player $j$.
    \item Consistency: if, for some value $v \in V , v_j = v$ for all honest players, then $o = v$.
\end{enumerate}

We refer to $o$ as $\mathcal{P}$’s output, and to each $o_i$ as player $i$’s output.
Agreement is reached on $\bot$ when it is not possible to agree on any other meaningful value in $V$.
\end{defn}

We remark that a network of nodes cannot always reach agreement on a meaningful piece of information.
In fact, if at the beginning of the protocol execution many players are in disagreement with each other, then none of the information advertised by the nodes can prevail.
For this reason the protocol must be designed in such a way that a high disagreement rate is detected, and the players output the symbol $\bot$ at the end of the protocol execution.

We assume that the players of the protocol form a network $\mathcal{N}$, that is essentially a graph that models the communication channels.
The network is made of $n$ nodes (the players), out of which less than $\frac{n}{3}$ are malicious or faulty.
The network graph is complete, which means that between any two distinct nodes there is a direct and private communication channel.

We will also require that the communications are carried out in a synchronous way, i.e. each node can access a common clock which triggers the start of every protocol step.
Finally, we assume that all communications are performed instantaneously.

Given this communication model, and the fact that honest nodes are supposed to send the same message to every node, throughout the paper we will use the terms ``send'' and ``broadcast'' interchangeably.

\subsection{Problem Description}
\label{probDesc}
In order to better motivate the generalization process that led to the design of protocol presented in this paper, let us introduce a model that we will show to encompass various practical problems, and an example situation for which solutions in literature (to the best of our knowledge) do not give satisfactory results.

Let $\mathcal{N}$ be a network, where each node $i\in \mathcal{N}$ has access to a Random Variable $\overline{X}^{(i)}$, where $\overline{X}^{(i)}=(X_{1}^{(i)},\dots,X_{m}^{(i)})$, for all  $c\in\{1,\dots,m\}$, where $X_c^{(i)}$ takes values in a discrete set $V_c$, $\mathbb{P}_{c}^{(i)}$ is its probability mass function, and $X_c^{(i)} = X_{c}^{(l)}$, for all $i,l\in \mathcal{N}$ and for all $c$.
Each node $i$ records ${O^{(i)}=\left(x_{1}^{(i)},\dots,x_{m}^{(i)}\right)}$, the \textit{observed values} given by the random variable.
The goal of our protocol is to allow the nodes in $\mathcal{N}$ to reach agreement on a vector of observed values.

We make a distinction between two kind of components: \emph{ambiguous} and \emph{unambiguous}.
A component $c$ 
is ambiguous if there exist two honest nodes $i,l\in \mathcal{N}$ who observe two distinct values $x_c^{(i)} \ne x_c^{(l)}$, otherwise we say that the component is unambiguous.

Essentially, ambiguous events cause disagreement on some vector components, even among honest nodes.
This means that nodes cannot agree on a vector as a whole, so we propose a \emph{leaderless} protocol in which the consensus process is carried out in parallel on each vector component.
By leaderless we mean that the consensus protocol is not carried on by evaluating the proposal of a single node (i.e. the proposal of a leader), instead, all nodes participating in the protocol equally work together.

Our setting is derived from a context in which some events are to be registered in a distributed ledger and no observer has a privileged point of view, contrary to the standard use-case in which blocks are proposed to the network by a single actor (e.g. by the node that first completes the creation of a block, or  by an elected temporary leader).
Of course different point of views may lead to different observations, so it is necessary to reconcile these differences in order to keep a common coherent ledger, even more so when taking into account the possibility of malicious behaviour from some observers.
A more specific application could be the timestamping  of events in a permissionless blockchain.
In this setting some users may be required to perform specific tasks within prescribed time limits, so our consensus protocol may be used to certify their good behaviour.
For example we may consider a blockchain that employs sharding and allocates block creation to miners in pre-determined time-slots (a technique employed by various existing platforms, e.g. EOS \cite{EOS2017} and Cardano \cite{kiayias2017ouroboros}): our BA protocol may be used to certify that blocks are indeed created during their prescribed time intervals, thus preventing attacks in which miners either delay block creation or pretend that previous blocks were late, leading to validation disputes and forks.

We now show, with the help of an example, the reason why in this model our leaderless and parallel approach allows for desirable outputs not readily attainable with existing alternatives.
Let $\mathcal{N}$ be a complete network with 4 honest nodes $j_1, j_2, j_3$ and $j_4$.
Each node $j_i$ in $\mathcal{N}$ observes $(X_{1}^{j_i}, X_{2}^{j_i}, X_{3}^{j_i}, X_{4}^{j_i})\in \mathbb{N}^4$, where $i\in\{1,2,3,4\}$.
Suppose that:
\begin{align*}
    O^{j_1}&=(9,2,8,4)\\
    O^{j_2}&=(9,2,7,1)\\
    O^{j_3}&=(9,3,8,1)\\
    O^{j_4}&=(0,2,8,1)
\end{align*}

If we decide to adopt a consensus protocol where a leader proposes to the network a vector to record, then the other nodes decide whether to accept it or not, a node will accept the leader's proposal only if the values it observed are the same as the ones advertised by the leader.

This can be done in at least two ways: \begin{itemize}
    \item \textit{the evaluation is performed on the whole vector}: this means that a node accepts the leader proposal only if the vector observed is the same as the vector the leader advertises. 
    In this particular context, the validity of the information that must be recorded translates to an accurate description of \emph{all} the events observed.
    Therefore, it is clear that, in our example, whoever the leader $l\in\{j_1, j_2, j_3, j_4\}$ is, the other nodes are not going to accept its proposal, since the vector observed by the leader differs from the vector of every other node.
    
    This would cause the output of the protocol to be the default vector $(\bot,\bot,\bot,\bot)$, which means that no meaningful data gets recorded. 
    
    \item \textit{the evaluation is performed in parallel on each component of the vector}: this means that a node can accept a component $c$ of the leader proposal only if the value it observes (associated to the $c$-th component) is the same as the one advertised by the leader. 
    In our simple example, whoever the leader $l\in\{j_1, j_2, j_3, j_4\}$ is, the nodes will agree on 3 components. In fact, a great majority of the network (3 out of 4 nodes) agrees on the values advertised by $l$ in 3 components of the vector, discarding the remaining component. 
    For instance, in our example, if the leader $l$ is the node $j_1$, then agreement will be reached on the vector $(9,2,8,\bot)$.
\end{itemize}

It is clear that the adoption of the second approach remarkably improves the first one. 
When the evaluation is performed on the whole vector, it is sufficient that the nodes do not agree with the leader in only one component to have all the components discarded. Whereas, the second approach allows the network to reach consensus and write in the ledger all the values proposed by the leader on which a majority of the nodes in the network agree.
This observation should convince the reader of the advantages given by the adoption of a consensus protocol which works in parallel on the vector components.

In our example, what is still undesirable is that the component that gets discarded is a component on which the majority of the network does agree.
Unluckily, the leader is part of the minority of the network which observed another value, and for this reason that component is discarded.
This emphasizes a weakness of leader-based consensus protocols when adopted to solve the consensus problem we are targeting.

With our proposal of a leaderless approach we aim to achieve a consensus protocol where the network listens to the opinion of more than a single node, thus agreeing on a vector where each component reflects the opinion of the majority.

In our example if all 4 nodes communicate to the other nodes their observed values, then agreement would be reached on the vector $(9,2,8,1)$, since a great majority of the nodes agrees on such values in each vector component.
This is the most desirable result since it is the outcome that one would expect from a group of nodes without hierarchy and whose opinions have equal value.

\section{Protocol components}
\label{prot-comp}
In this section we generalize the Binary Byzantine Agreement protocol of \cite{micali2016byzantine} and the $(n,t)$-Graded Consensus protocol of \cite{chen2016algorand}, extending them to the multidimensional case.
These sub-protocols are the two building blocks that we will use to define our MBA protocol.

\subsection{Notation}
In this paper we will use the useful notation $\#_i^s(v)$ adopted by Micali in \cite{chen2016algorand}
(or just $\#_i(v)$ when $s$ is clear) to represent the number of players from which $i$ has received the message $v$ during step $s$, counting also his own message if he has sent a message $v$ during step $s$.

Assuming that in each step $s$ a player $i$ receives exactly one message from each player $j$, if the number of players is $n$, then $\sum_{v} \#_i^s(v)=n \quad \forall i, s$.
During the protocol execution honest players should send only one message per step, so, if player $i$ receives from player $j$ two contrasting messages, then $i$ discards both so they are not included in the count when computing $\#_i()$.
Two identical messages are instead counted as one, and messages that are not properly formatted are discarded as well, so only \emph{valid} messages are considered and counted, and $\sum_{v} \#_i^s(v)\leq n \quad \forall i, s$.

Similarly to $\#_i^s(v)$, when the exchanged message is an $m$-dimensional vector of strings $\mathbf{v}=(v_1,\dots,v_m)$, we define $\#_i^s(v,c)$,  for $c\in\{1,2,\dots,m\}$ (or just $\#_i(v,c)$ when $s$ is clear) as the number of players from which $i$ has received during step $s$ a vector of strings $\mathbf{v}$ such that $v_c=v$.

When the messages exchanged by each player $j$ are $m$-dimensional vectors of strings $\mathbf{v_j}=(v_{j,1},v_{j,2},\dots,v_{j,m})\in (V\cup\{\bot\})^m$, we also define the concept of \emph{$c$-agreement}, where $c\in\{1,2,\dots,m\}$ is a specific component of the vectors.
We say that the players reached \emph{$c$-agreement} when there exists $v \in V \cup \{\bot\}$ such that for every honest player $j$, $v_{j,c}=v$.

When $c$-agreement is reached on all the components of the vector, we have that for all honest players $i,j$, $\mathbf{v_i}=\mathbf{v_j}$, hence also agreement is reached.

Finally, we write $\mathbf{0}$ and $\mathbf{1}$ to represent the vectors with in each component 0 and in each component 1, respectively.

\subsection{Multidimensional Binary Byzantine Agreement Protocol}
\label{MBBA subsection}
We now introduce a multidimensional extension of the binary byzantine agreement protocol \emph{BBA} presented by Micali in \cite{micali2016byzantine} and we will call it \emph{Multidimensional Binary Byzantine Agreement} (\emph{MBBA}).
The protocol uses a cryptographic hash function $H$ modeled as a random oracle, and we order its outputs (which are bit strings) with a standard lexicographic order.
We use also a digital signature algorithm with unique signature, and we denote with $\mult{SIG}_i(x)$ the unique signature on the bit string $x$ computed by player $i$.
That is, let $\sigma$ be any signature that verifies against $x$ and the public key of $i$ computed by any party in polynomial time, then the probability that $\sigma \ne \mult{SIG}_i(x)$ is negligible.
Note that this means that even with the private key it is infeasible to compute a different signature for the same message.
We suppose that the players' public keys are known by everyone, so every player can verify any signature.

The protocol is an iterated procedure, where at each iteration three steps are performed.
To track the iterations, it uses a counter $\gamma$ representing how many times the 3 steps loop has been performed during a single protocol execution. 
At the beginning of an MBBA Protocol execution, $\gamma=0$.
Also, the protocol requires a minimal setup: a common random string $r$ independent of the nodes' keys.

\begin{protocol}{\emph{MBBA}}
\vspace{5pt}
    Each player $i$ privately knows a bit vector $\mathbf{b}_i=(b_{i,1},b_{i,2},\dots,b_{i,m})$ and locally saves a $m$-bits vector $\mathbf{f}_i=\mathbf{0}$ .
    \begin{itemize}
        \item \textit{EXIT CHECK}. If $\mathbf{f}_i=\mathbf{1}$, player $i$ sends $\mathbf{b}_i\star=(b_{i,1},\dots,b_{i,m})\star$, outputs $o_i=\mathbf{b}_i$ and HALTS. 
        \item \textit{STEP 1}. [Coin-Fixed-To-0 Step] Each player $i$ sends $\mathbf{b}_i$. For all $c\in\{1,2,\dots,m\}$ s.t. $f_{i,c}=0$:
        \begin{enumerate}
            \item if $\#_i^1(0,c)>\frac{2}{3}n$, then $i$ sets $b_{i,c}=0$, sets $f_{i,c}=1$, and performs the EXIT CHECK.
            \item if $\#_i^1(1,c)>\frac{2}{3}n$, then $i$ sets $b_{i,c}=1$.
            \item Else $i$ sets $b_{i,c}=0$.
        \end{enumerate}
        
        \item \textit{STEP 2}. [Coin-Fixed-To-1 Step] Each player $i$ sends $\mathbf{b}_i$. For all $c\in\{1,2,\dots,m\}$ s.t. $f_{i,c}=0$:
        \begin{enumerate}
            \item if $\#_i^2(1,c)>\frac{2}{3}n$, then $i$ sets $b_{i,c}=1$, sets $f_{i,c}=1$, and performs the EXIT CHECK.
            \item if $\#_i^2(0,c)>\frac{2}{3}n$, then $i$ sets $b_{i,c}=0$.
            \item Else $i$ sets $b_{i,c}=1$.
        \end{enumerate}
        
        \item \textit{STEP 3}. [Coin-Genuinely-Flipped Step] Each player $i$ sends the signature $s_i=\mult{SIG}_i(r\| \gamma)$  and $\mathbf{b}_i$. For all $c\in\{1,\dots,m\}$ s.t. $f_{i,c}=0$:
        \begin{enumerate}
            \item if $\#_i^3(0,c)>\frac{2}{3}n$, then $i$ sets $b_{i,c}=0$.
            \item if $\#_i^3(1,c)>\frac{2}{3}n$, then $i$ sets $b_{i,c}=1$.
            \item Else, letting $P_{i}$ be the set of players $j$ who sent $i$ a valid message in STEP 3, $i$ computes $k=H(\min_{j \in P_i}H(s_j))$ and sets $b_{i,c}=k_c$ where $k_c$ is the $c$-th bit in $k$.
        \end{enumerate}
        Player $i$ increases the counter $\gamma$ by 1, and returns to STEP 1.
    \end{itemize}
\end{protocol}

The symbol $\star$ is applied to each message which must be considered final.
That is, if a player $i$ receives a message $v\star$ from a node $j$, then $i$ must pretend that, in every following step, $j$ will send the same message $v$.

\begin{remark}
When $m=1$ the protocol \emph{$m$-Dimensional BBA} is exactly the protocol \emph{BBA} described in \cite{micali2016byzantine}.
\end{remark}

\begin{thm}
\label{mBBA}
Whenever $n \ge 3t + 1$, the \emph{$m$-Dimensional BBA} protocol is an $(n, t)$-Byzantine Agreement Protocol with soundness 1.
\end{thm}

The proof of such theorem follows the proof of the analogous theorem related to the protocol BBA \cite{micali2016byzantine}. However, the parallelization required some adjustments to the protocol itself, therefore the proof must be adapted accordingly.

We first prove some Lemmas which will lead us to a proof of \Cref{mBBA}.

\begin{lem}
\label{lem2}
If, at some step, an honest player $i$ sets $f_{i,c}=1$, then $c$-agreement will hold at the end of the step.
\end{lem}
\begin{proof}
First, note that an honest player $i$ can set $f_{i,c}=1$ only during a Coin-Fixed-To-0 Step or a Coin-Fixed-To-1 Step, so STEP 3 is not taken under consideration.

Assume that $i$ fixes $b_{i,c}=0$ and $f_{i,c}=1$ in a Coin-Fixed-To-0 Step.
This means that $\#_i^1(0,c)>\frac{2}{3}n$, so more than $\frac{1}{3}n$ honest players have sent 0 at the start of such step (in fact the malicious players are less than $\frac{1}{3}n$).
Since honest players send to everyone the same message, then $\#_j^1(0,c) > \frac{1}{3}n$ for each other honest player $j$.
For such players two mutually exclusive cases may occur :
\begin{enumerate}
    \item $\#_j^1(0,c) >\frac{2}{3}n$: in this case $j$ sets $b_{j,c} = 0$ (and also $f_{j,c}=1$) in sub-step 1 of STEP 1.
    \item $\frac{1}{3}n < \#_j(0,c) \le \frac{2}{3}n$: in this case $j$ must enter the third sub-step of STEP 1 hence sets $b_{j,c} = 0$.
\end{enumerate}
In any case we have that every honest player $j$ sets $b_{j,c} = 0$, thus $c$-agreement holds on 0 at the end the Coin-Fixed-To-0 Step.

To conclude our proof, note that a symmetric argument shows that $c$-agreement holds on 1 at the end of a Coin-Fixed-To-1 Step in which an honest player $i$ sets $b_{i,c}=1$ and $f_{i,c}=1$.
\end{proof}

\begin{lem}
\label{lem1}
For each component $c\in\{1,2,\dots,m\}$,
if at some step $c$-agreement holds, then it continues to hold in the next steps.
\end{lem}
\begin{proof}
Assume that, for some component $c$, the players reached $c$-agreement at the end of STEP $s$.
We want to show that in each subsequent step $c$-agreement still holds.
We assume that $c$-agreement has been reached on 0. A similar analysis can be done in case $c$-agreement is reached on 1. 

Let us consider the three possible options for the step $s$:
\begin{itemize}
    \item $s$ is a STEP 3, so the next step is STEP 1.
    
    At the beginning of STEP 1 each honest player $i$ sends its vector $\mathbf{b}_i$.
    Since agreement has been reached on component $c$ during the previous step, for all honest $i$, $b_{i,c}=0$, so $\#_i^1(0,c)> \frac{2}{3}n$ given that the honest players are more than $\frac{2}{3}n$.
    This means that in STEP~1, for the component $c$, each honest player $i$ enters in the first sub-step, so $c$-agreement still holds, since the component $c$ is left unchanged.

    \item $s$ is a STEP 1, so the next step is STEP 2.
    
    At the beginning of STEP 2 each honest player $i$ sends its vector $\mathbf{b}_i$.
    As in the previous case, since agreement has been reached on component $c$ during the previous step, for all honest $i$ we have that $b_{i,c}=0$ and $\#_i^2(0,c) >\frac{2}{3}n$.
    This means that in STEP 2, for the component $c$, each honest player $i$ enters in the second sub-step, so $c$-agreement still holds since the component $c$ is left unchanged.

    \item $s$ is a STEP 2, so the next step is STEP 3.
    
    At the beginning of STEP 3 each honest player $i$ sends its vector $\mathbf{b}_i$.
    Again, since agreement has been reached on component $c$ during the previous step, for all honest $i$ we have that $b_{i,c}=0$ and $\#_i^3(0,c) >\frac{2}{3}n$.
    This means that in STEP 3, for the component $c$, each honest player $i$ enters in the first sub-step, so $c$-agreement still holds since the component $c$ is left unchanged.
\end{itemize}

Thus, if $c$-agreement holds at some step, it will keep holding during the next step, and so on until every player halts.
\end{proof}

\begin{lem}
\label{lem3}
Let $h>\frac{2}{3}$ be the ratio of honest players in the network.
If, at the start of an execution of STEP 3, no player has halted, i.e. agreement has not been reached yet on a bit vector, then, being $l$ the number of vector components $c$ on which $c$-agreement has not been reached, the players will be in agreement at the end of this step with probability at least $h(\frac{1}{2})^l > \frac{2}{3}(\frac{1}{2})^l$. 
\end{lem}

\begin{proof}
Let $\gamma$ be the current value of the counter, and let $P_i$ be the set of players from which $i$ has received a valid message at the beginning of STEP 3.
By the uniqueness property of the underlying digital signature scheme, $i$ can compute $\mathbf{k}=H({min}_{j \in P_i}H(\mult{SIG}_j(r\|\gamma)))$, and then $k_1,\dots,k_m$ are well
defined.

Note that the selection of the player $p$ whose hashed digital signature is minimal is a random selection with uniform distribution, under the assumption that $H$ is a \emph{random oracle}.
This means that $p$ will be an honest player with probability $h>\frac{2}{3}$, and in this case it will send its message to every player.
In particular all the honest players will receive $\mult{SIG}_{p}(r\|\gamma)$, from which the same values $k_1,\dots,k_m$ will be computed by all the honest players who perform sub-step 3 of STEP 3.

Let $\{c_i\}_{i=1,\dots,l}$ be the set of components of the bit vector on which $c_i$-agreement does not hold, and let us assume that the player $p$ is honest.
For each component $c_i$ of the bit vector, notice that it is impossible that some honest players perform sub-step 1 and some sub-step 2 of STEP 3.
In fact, being $t<\frac{1}{3}n$ the number of malicious nodes, if a node $i$ has received more than $\frac{2}{3}n$ messages for 1 and a node $j$ has received more than $\frac{2}{3}n$ valid messages for 0, then $i$ has received at least $\frac{2}{3}n-t>\frac{1}{3}n$ messages for 1 from honest nodes, and such messages have reached also $j$.
However, $j$ has received at least $\frac{2}{3}n$ valid messages for 0, which is a contradiction since $\frac{2}{3}n+\frac{2}{3}n-t>n$, and $j$ cannot receive more than $n$ valid messages from distinct nodes.

Therefore there are five exhaustive cases that must be considered and may lead the honest players to $c_i$-agreement:
\begin{itemize}
    \item \emph{All honest players update their $c_i$th component according to sub-step 1 of STEP 3.}
    
    In this case $c_i$-agreement hols on 0.
    
    \item \emph{All honest players update their $c_i$th component according to sub-step 2 of STEP 3.}
    
    In this case, $c_i$-agreement holds on 1.
    
    \item \emph{All honest players update their $c_i$th component according to sub-step 3 of STEP 3.}
    
    In this case, at the end of Step 3, $c_i$-agreement holds on $k_{c_i}$ (we assume $p$ is honest).
    
    \item \emph{Some honest players update their $c_i$th component according to sub-step 1 of STEP 3 and all the others according to sub-step 3 of STEP 3.}
    
    The honest players updating the value according to sub-step 1 will set the $c_i$th component of their vector to 0, while the ones updating their $c_i$th component according to sub-step 3 will set it to $k_{c_i}$ which is 0 with probability $\frac{1}{2}$.
    This means that $c_i$-agreement is reached on 0 with probability $\frac{1}{2}$.
    
    \item \emph{Some honest players update their $c_i$th component according to sub-step 2 of step 3 and all the others according to sub-step 3 of STEP 3.}
    
    The honest players updating the value according to sub-step 1 will set the $c_i$th component of their vector to 1, while the ones updating their $c_i$th component according to sub-step 3 will set it to $k_{c_i}$ which is 1 with probability $\frac{1}{2}$.
    This means that $c_i$-agreement is reached on 1 with probability $\frac{1}{2}$.
\end{itemize}

When the player $p$ is honest, we can assume that the values $k_c$ are chosen randomly and independently, under the assumption that $H$ is a random oracle.
Hence at the end of STEP 3 the players will reach $c$-agreement for all values of $c \in \{1,\dots,m\}$, which means agreement, with probability at least $(\frac{1}{2})^l$.
Thus, given that the probability of having the player $p$ honest is $h > \frac{2}{3}$, we can conclude that, anytime the players reach STEP 3 of the protocol, before the end of such step they will be in agreement with probability at least $h(\frac{1}{2})^l > \frac{2}{3}(\frac{1}{2})^l$.
\end{proof}

We now can prove \Cref{mBBA}.
\begin{proof}
We prove the following properties that characterize a Byzantine agreement protocol with soundness $\sigma=1$.
\begin{enumerate}
    \item \label{Halt} All honest players HALT with probability 1.
    
    If at the beginning of STEP 3 the players are not in $c$-agreement over $l$ components, with probability  at least $\frac{2}{3} (\frac{1}{2})^l > 0$ they will be in agreement at the end of that step, hence with the growing of the number of STEP 3 executions the probability to reach agreement converges to 1. Note that, at every STEP 3 execution, the number of components not agreed upon can not increase, so the probability to end the protocol in the next STEP 3 execution does not decrease.
    
    Once agreement is reached, the honest players will HALT in the following 2 steps since it will finalize the zeroes in STEP~1 and the ones in STEP~2 (updating the locally saved vector $\mathbf{f}$ with ones corresponding to the finalized components).
    
    \item $o_i = o_j$ for all honest players $i$ and $j$.
    
    This is true because by point \ref{Halt} all honest players HALT, thus they have $\mathbf{f}=\mathbf{1}$ and by applying \Cref{lem2} to every component we can conclude that they are in agreement, so $o_i = o_j$.
    
    \item If the initial value of every honest player $i$ is a vector $\mathbf{b}_i=\mathbf{b}$, then $o_i=\mathbf{b}$ for every honest player $i$.
    
    Note that $c$-agreement holds at the start of the protocol for every $c$, and by \Cref{lem2} it will continue to hold.
    So every honest player sends the same vector $\mathbf{b}$ at the beginning of each round, so $\#_i(b_c, c) > \frac{2}{3}n$ for every $c$ and for every honest $i$.
    It is exhaustive to consider the following two cases:
    \begin{itemize}
        \item if $\mathbf{b}$ is the vector of all zeros, then all honest players enter sub-step 1 of STEP 1 and when they verify the EXIT CHECK, once they have updated all their components, they will halt setting $o_i=\mathbf{b}$.
        
        \item Otherwise during STEP 1, for all $c \in \{1,2,\dots,m\}$ if $b_{i,c}=0$ they will set $b_{i,c}=0$, and update $\mathbf{f}_i$ by setting $f_{i,c}=1$, but will not halt since for some $c$ we have $b_{i,c}=1$, thus $f_{i, c} = 0$.
        Each of these components $c$ will be finalized during STEP 2 when it will be set $b_{i,c}=1$ and $f_{i,c}=1$.
        In this case, once the last coordinate is updated the EXIT CHECK will be verified, and every honest player $i$ will output $o_i=\mathbf{b}$.
    \end{itemize}
\end{enumerate}
\end{proof}

\subsection{Multidimensional Graded Consensus Protocol}
\label{MGC subsection}
The notion of \emph{Graded Consensus} protocol, introduced by Micali in \cite{chen2016algorand}, is much weaker than Byzantine agreement but allows the protocol players to gain some information about the distribution of the input values possessed by the network participants.

We now provide the definition of an extension of the concept of $(n,t)$-Graded Consensus to the multidimensional case. 

\begin{defn}[\emph{$m$-Dimensional $(n,t)$-Graded Consensus Protocol}] \label{mGradedConsensus}
	Let $\mathcal{P}$ be a protocol in which the set of all players is common knowledge, and each player $i$ privately knows an arbitrary initial vector of messages $\mathbf{v_i'}=(v_{i,1}',v_{i,2}',\dots,v_{i,m}')$ where  $v_{i,h}'\in V\cup\{\bot\}$.
	
	We say that $\mathcal{P}$ is an \emph{$m$-dimensional $(n, t)$-graded consensus protocol} if, in every execution with $n$ players of which at most $t$ are malicious, every honest player $i$ halts outputting a vector of value-grade pairs $o_i=(\mathbf{v}_i,\mathbf{g}_i)=((v_{i,1},g_{i,1}),(v_{i,2},g_{i,2}),\dots,(v_{i,m},g_{i,m}))$ where $g_{i,c}\in\{0,1,2\}$ and $v_{i,c}\in V\cup\{\bot\}$ for every $i$ and $c$, so as to satisfy the following three conditions:
	\begin{enumerate}
		\item For all honest players $i$ and $j$, for all $c \in \{1,\dots,m\}$, we have that $\lvert g_{i,c}-g_{j,c} \rvert \leq 1$.
		\item For all honest players $i$ and $j$, for all $c \in \{1,\dots,m\}$ and for all positive $ g_{i,c},g_{j,c}$ we have $v_{i,c}=v_{j,c}\ne\bot$.
		\item
		If $v_{1,c}'=v_{2,c}'=\dots=v_{n,c}'=v_c$ for some value $v_c \in V\cup \{\bot\}$, then for all honest players the output component $c$ is $(v_{i,c},g_{i,c})=(v_c,2)$ if $v_c\ne\bot$, $(v_{i,c},g_{i,c})=(\bot,0)$ if $v_c=\bot$ .
	\end{enumerate}
\end{defn}

\begin{remark}
A 1-dimensional $(n,t)$-graded consensus protocol is an $(n,t)$-graded consensus protocol according to the definition in \cite{chen2016algorand}.
\end{remark}

\begin{remark} \label{rem3}
An immediate consequence of condition 3 is that if the initial vectors of the players are equal $\mathbf{v}_1'=\dots=\mathbf{v}_n'=\mathbf{v}=(v_1,\dots,v_m)\in (V\cup \{\bot\})^m$ then also the outputs will be the same, where $(v_{i,c},g_{i,c})=(v_{i,c}',2)$ when $v_{i,c}\ne\bot$ and $(v_{i,c},g_{i,c})=(\bot,0)$ when $v_{i,c}'=\bot$.
\end{remark}

In the same way we described the $m$-Dimensional BBA we define a multidimensional extension of the GC protocol, \emph{Multidimensional Graded Consensus} (\emph{MGC}).\\
 
\begin{protocol}{\emph{MGC}}
	\vspace{5pt}
	Each player $i$ privately knows some value $\mathbf{v}'_i=({v'}_{i,1},{v'}_{i,2},\dots,{v'}_{i,m})$.\\
	\begin{itemize}
		\item \textit{STEP 1.} Each player $i$ sends $\mathbf{v}'_i$ to all players.
		
		\item \textit{STEP 2.} Each player $i$ sends to all players the vector $\mathbf{\tilde{v}_i}=(\tilde{v}_{i,1},\tilde{v}_{i,2},\dots,\tilde{v}_{i,m})$ where $\tilde{v}_{i,c}=v$ if and only if $\#_i^1(v,c)\ge \lfloor\frac{2n}{3}\rfloor+1$, otherwise $\tilde{v}_{i,c}=\bot$.
		
		\item \textit{OUTPUT DETERMINATION.} Each player $i$ outputs the vector of pairs $(\mathbf{v}_i,\mathbf{g}_i)=((v_{i,1},g_{i,1}),(v_{i,2},g_{i,2}),\dots,(v_{i,m},g_{i,m}))$ where for all $c \in \{1,\dots,m\}$, $(v_{i,c},g_{i,c})$  is computed as follows:	\begin{itemize}
			\item If, for some $x \ne \bot$, $\#_i^2(x,c)\geq \lfloor\frac{2n}{3}\rfloor+1$, then $(v_{i,c},g_{i,c})=(x,2)$.
			\item If, for some $x \ne \bot$, $\#_i^2(x,c)\geq \lfloor\frac{n}{3}\rfloor+1$, then $(v_{i,c},g_{i,c})=(x,1)$.
			\item Else, $(v_{i,c},g_{i,c})=(\bot,0)$.
		\end{itemize}
	\end{itemize}
\end{protocol}

\begin{thm} \label{mGC_thm}
	If $t=\lfloor\frac{n-1}{3}\rfloor$ then MGC is an $m$-Dimensional $(n,t)$-graded consensus protocol.
\end{thm}

We have extended the Graded Consensus protocol GC introduced in \cite{chen2016algorand} executing simultaneously $m$ instances of GC protocol. 

The protocol GC is derived by the \emph{Gradecast Protocol} whose properties are proved in \cite{feldman1997optimal}.
An explicit proof for the properties of the GC protocol would be enough to be sufficiently assured that the properties of MGC hold. However an explicit proof of GC properties is not provided in \cite{chen2016algorand}, therefore, for sake of clarity, we will prove \Cref{mGC_thm}.

\begin{proof}[Proof of \Cref{mGC_thm}]
    The analysis is performed on a generic component $c$, and the final result is a consequence of the properties holding on every component. \\
    \begin{enumerate}
        \item We first prove that it is impossible for two honest players $i$ and $j$ to end the protocol with $c$-th output components $(v_{i,c},g_{i,c})$ and $(v_{j,c},g_{j,c})$ such that $\lvert g_{i,c}-g_{j,c} \rvert = 2$.
	
        Let us assume $g_{i,c}=0$ (hence $v_{i,c}=\bot$) and $g_{j,c}=2$ (hence $v_{j,c}\ne\bot$).
        This means that, at the end of STEP 2, $\#_j^2(v_j,c)\ge \lfloor\frac{2n}{3}\rfloor+1$.
        Out of these messages, the honest players have sent at least
        $ \lfloor\frac{2n}{3}\rfloor+1 - \lfloor\frac{n-1}{3}\rfloor > \frac{2n}{3} - \frac{n}{3} = \frac{n}{3} $ of them.
        Note that $\lfloor\frac{n}{3}\rfloor+1$ is the smallest integer strictly greater than $\frac{n}{3}$, so the honest players have sent at least $\lfloor\frac{n}{3}\rfloor+1$ messages.
        Since the messages sent by honest players are received both by $j$ and by $i$, $\#_i^2(v_j,c)\ge \lfloor \frac{n}{3} \rfloor +1$ hence $g_{i,c}$ cannot be 0.
    
        \item We now prove that if $i, j$ are honest players and $g_{i,c},g_{j,c}>0$, then $v_{i,c}=v_{j,c}$.
	
	    Assume that $g_{i,c},g_{j,c}>0$ and $v_{i,c}\ne v_{j,c}$.
	    This means that $\#_i^2(v_i,c)\ge \lfloor\frac{n}{3}\rfloor+1$ and $\#_j^2(v_j,c)\ge \lfloor\frac{n}{3}\rfloor+1$.
	    Since there are at most $\lfloor\frac{n-1}{3}\rfloor$ malicious players, at least one of each of these sets of messages comes from an honest player.
	    This means that, at the beginning of STEP 2, at least two distinct honest players $h,k$ have received $\#_h^1(v_i,c)\ge \lfloor\frac{2n}{3}\rfloor+1$ and $\#_k^1(v_j,c)\ge \lfloor\frac{2n}{3}\rfloor+1$ messages.
	    This is impossible since honest players sent the same messages both to $h$ and $k$, and the number of malicious players is $\hat{t}\le t<\frac{n}{3}$.
	    The malicious players may have sent different messages to $h$ and $k$, but they can have sent no more than $t = \lfloor \frac{n-1}{3} \rfloor$ messages to each of the honest players.
	    Considering all the distinct messages received by either $h$ or $k$ at the end of STEP 2, we have that at most $n-\hat{t}$ of them have been sent by honest players, and at most $2\hat{t}$ by malicious players, so there are at most $n+\hat{t}\le n+t$ distinct messages.
	    However, considering the messages received for the values $v_i$ and $v_j$, we have that $h$ and $k$ have received at least $2(\lfloor\frac{2n}{3}\rfloor + 1) > 2 (\frac{2n}{3}) = n + \frac{n}{3} >  n + \lfloor\frac{n-1}{3}\rfloor \ge n + t$ messages, which contradicts the fact that the number of messages received by $h$ and $k$ cannot exceed $n+t$.
	    
	    \item 
	    \begin{enumerate}
	        \item We now prove that if $v'_{i,c}=v_c\ne\bot$ $\forall i \in\{1,\ldots,n\}$ for some value $v_c$, then for all honest players the output is $(v_{i,c},g_{i,c})=(v_c,2)$.
	        
	        This is true because at the end of STEP 1 each honest player broadcasts $v_c$.
    	    Note that the honest players are at least $n - \lfloor\frac{n-1}{3}\rfloor > n - \frac{n}{3} = \frac{2n}{3}$, and since there are an integral number of them they are at least $\lfloor\frac{2n}{3}\rfloor+1$.
    	    This means that for each honest player $i$, $\#_i^1(v_c,c)\ge\lfloor\frac{2n}{3}\rfloor+1$ thus each honest player in STEP 2 broadcasts $v_c$.
    	    Again for each honest player $i$ must be $\#_i^2(v_c,c)\ge\lfloor\frac{2n}{3}\rfloor+1$, and this implies $(v_{i,c},g_{i,c})=(v_c,2)$.
    	    
    	    \item Finally, we prove that if  $v'_{i,c}=\bot$ $\forall i \in\{1,\ldots,n\}$, then all honest players output $(v_{i,c},g_{i,c})=(\bot,0)$.
    	    
    	    In this case at the end of STEP 1 each honest player broadcasts $\bot$.
    	    This means that for each honest player $i$, $\#_i^1(\bot,c)\ge\lfloor\frac{2n}{3}\rfloor+1$ thus there cannot exist a value $v_c\ne \bot$ such that $\#_i^1(v_c,c)\ge \lfloor\frac{2n}{3}\rfloor+1$ (otherwise the number of messages considered by the honest player $i$ would exceed $n$).
    	    Hence $i$ will send the message with $\bot$ in the $c$-th component at the end of STEP 2.
    	    Again, for each honest player $i$, $\#_i^2(\bot,c)\ge\lfloor\frac{2n}{3}\rfloor+1$ thus there cannot exist a value $v_c\ne \bot$ such that $\#_i^2(v_c,c)\ge \lfloor\frac{n}{3}\rfloor+1$ and this implies $(v_{i,c},g_{i,c})=(\bot,0)$.
	    \end{enumerate}
    \end{enumerate}
\end{proof}

\begin{remark}
Note that, if for some honest player $i$ we have that $g_{i,c}=2$, then, by Property~1 of \Cref{mGradedConsensus}, for each honest player $j$ we have that $g_{j,c}\ge 1$.
Therefore, by Property~2 of \Cref{mGradedConsensus}, since  $g_{k,c}\ge 1$ for each honest player $k$, we have that $v_{k,c}=v_c \ne \bot$.
\end{remark}

\section{Multidimensional Byzantine Agreement}
\label{MBA section}
We now combine the MGC and MBBA protocols to create a Multidimensional Byzantine agreement protocol MBA that allows the players in a synchronous network to reach agreement on an arbitrary vector of values.

\begin{protocol}{MBA}
\vspace{5pt}
The initial value of each player $i$ is a vector $v_i' \in (V\cup\{\bot\})^m$.
    \begin{itemize}
        \item STEPS 1 and 2. Each player $i$ executes the first two steps of the $m$-dimensional GC so as to compute a value-grade pair vector $\mathbf{(v_i,g_i)}$ (which will be referred to as the output of STEP 2).  
        
        \item STEP 3. Each player $i$ executes $m$-dimensional BBA with initial vector $\mathbf{b}_i$ where for all $c \in \{1,\dots,m\}$ $b_{i,c}=0$ if $g_{i,c}=2$, $b_{i,c}=1$ otherwise. (The bit vector obtained will be referred to as the output of STEP 3).
        
        \item OUTPUT DETERMINATION. Each player $i$ outputs $\mathbf{o_i}$ where $o_{i,c}=v_i(\ne\bot)$ if $b_{i,c}=0$ and $v_{i,c}=\bot$ if $b_{i,c}=1$.
    \end{itemize}
\end{protocol}

Let us now prove that the MBA protocol is indeed a Byzantine Agreement protocol, starting with a useful lemma.

\begin{lem}
\label{lemAg}
If, for some component $c$ of the output $\mathbf{b}$ of STEP 3, the honest players get $b_c=0$, then they reach $c$-agreement on a value $v\ne\bot$ in the $c$-th component of the output of MBA.
\end{lem}
\begin{proof}
At the end of STEP 3, before determining the output of the protocol MBA,
if each honest player $j$ gets $b_c=0$ ($\mathbf{b}$ is common to all honest players since $m$-dimensional BBA is a Byzantine Agreement protocol),
then it means that at the beginning of STEP 3 at least one honest player $i$ had  $b_{i,c}=0$ (otherwise by the consistency property of Byzantine Agreement protocols they would output 1).
This means that at the end of STEP 2 $g_{i,c}=2$.

By Properties 1 and 2 of \Cref{mGradedConsensus}, for any other honest player $j$ we respectively get $g_{j,c}\ge 1$ and $v_{j,c}=v_{i,c}$.
Then when $j$ computes the $c$ component of output of STEP 3, since $b_c=0$ we have that $j$ sets $o_{j,c}=v_{i,c}$, as all the other honest players will do.
Thus $c$-agreement is reached. 
\end{proof}

\begin{thm}
  Whenever $n\ge3t+1$ MBA is an $(n,t)$-Byzantine Agreement protocol with soundness 1.  
\end{thm}

\begin{proof}
We have already proven that MBBA halts with probability 1. Since MGC is not an iterative protocol, after its 2 steps it will halt with probability 1. Therefore, MBA halts with probability 1 as well.
We must prove the Consistency and Agreement properties.
\begin{itemize}
    \item \emph{Consistency}: we assume that, for each player $i$, the initial vector is ${\mathbf{v_i'}=\mathbf{v}\in (V\cup\{\bot\})^m}$.
    By Property 3 of $m$-dimensional Graded Consensus, at the end of the second step of protocol MBA the output $(\mathbf{v}_i,\mathbf{g}_i)=((v_{i,1},g_{i,1}),\dots,(v_{i,m},g_{i,m}))$ of any honest player $i$ has, for all $c \in \{1,\dots,m\}$, $(v_{i,c},g_{i,c})=(v_{i,c}',2)$ if $v_{i,c}'\ne\bot$ and $(v_{i,c},g_{i,c})=(\bot,0)$ if $v_{i,c}'=\bot$.
    Accordingly, the honest players will agree on the initial bit vector of STEP 3 of MBA: in particular they will set $b_{i,c}=0$ if $v_{i,c}\ne\bot$, $b_{i,c}=1$ if $v_{i,c}=\bot$.
    By the Agreement property of $m$-dimensional BBA we obtain that the agreed upon bit vector of STEP 3 will be the same for all honest players and equal to the initial vector of $m$-dimensional BBA.
    Hence, by the MBA protocol definition, the output of the protocol MBA of each honest player $i$ will be $\mathbf{o}_i=\mathbf{v}$, the common initial vector.
    
    \item \emph{Agreement}: since $m$-dimensional BBA is a Byzantine Agreement protocol, all the honest players will end STEP 3 with the same bit vector $\mathbf{b}$.
    Each honest player $i$ will compute $\mathbf{o}_i$ in the following way for each component $c$ of $\mathbf{b}$:
    \begin{enumerate}
        \item either $b_c=0$ for all honest players: in this case $c$-agreement on the outputs holds thanks to \Cref{lemAg};
        \item otherwise $b_c=1$ for all honest players: in this case all players will set $out_{i,c}=\bot$, so $c$-agreement on the output still holds.
    \end{enumerate}
    Since $c$-agreement holds for every component $c$, we can state that Agreement holds.
    \end{itemize}
\end{proof}

Note that the protocol MBA is the multidimensional version of the protocol \emph{BA*} described in \cite{chen2019algorand}.

\subsection{A Probability Game}
\label{PG}
\label{game}
We consider the following game that will be used to model the evolution of the component finalization process in the MBBA protocol, and thus the MBA protocol.

In particular we want compute the probability distribution associated to the number of steps necessary to win this game.
From that, we retrieve the probability distribution associated to the number of MBBA iterations necessary to end the \emph{MBA} Protocol.

The game is the following: we have $n$ coins which flip heads with probability $\pi$, at each step we flip the coins and discard the ones which flipped heads, then we carry on with the others until there are no coins left.
So, in the first step we flip all $n$ coins, then we discard the $h_1$ coins which flipped heads, in the second step we flip the remaining $n - h_1$ coins and so on.

We now compute the probability distribution associated to the number $\chi_{n,\pi}$ of steps required to end the game.

The probability that a coin flips heads at least once in $w$ steps is $1-(1-\pi)^w$, hence, being the coin flips independent, the probability that all coins flip heads at least once 
is $(1-(1-\pi)^w)^n$.
This means that
$$  P(\chi_{n,\pi}>w)=1-(1-(1-\pi)^w)^n,$$
and from that we can compute
\begin{align*} 
P(\chi_{n,\pi}=w)&=P(\chi_{n,\pi}>w-1)-P(\chi_{n,\pi}>w)\\
&=(1-(1-\pi)^{w})^n - (1-(1-\pi)^{w-1})^n.
\end{align*}
This defines the probability distribution associated to the number of steps to end this game played with $n$ coins.

Note that in the Coin-Genuinely-Flipped step of the MBBA protocol the bits in the ambiguous components (i.e. the ones where some honest player has less than $\frac{2}{3}n$ confirmations) are randomly flipped, while the components in agreement are left untouched.
Moreover note that these flips cause the value of the ambiguous components to match the one held by the other honest players (thus reaching agreement) with probability greater than $h \cdot 2^{-l}$ (see \Cref{lem3}), and that we end the protocol when all components are in agreement, so the probability distribution $\chi_{n,\frac{h}{2}}$ gives an upper bound on the distribution of the number of Coin-Genuinely-Flipped steps necessary to end the MBBA protocol.
To connect more directly with the analysis of \Cref{lem3}, note that the game above ends in one step with probability $\pi^n$ and that $(\frac{h}{2})^l \le h \cdot 2^{-l}$, so this lower bound in the probability translates in an upper bound in the number of steps necessary to conclude the protocol.
\subsection{Number of Steps}
\label{mainT}
\begin{thm}
    The distribution of the number of MBBA iterations required to end the MBA protocol run is upper bounded by the random variable $ 1 + \chi_{l,\frac{h}{2}}$ where
    \begin{itemize}
        \item $l$ is the number of ambiguous components;
        \item $h$ is the ratio of honest nodes in the network;
        \item $\chi_{l,\frac{h}{2}}$ is the random variable described in \Cref{game}.
    \end{itemize}  
\end{thm}

\begin{proof}
We recall that by \Cref{lem1} if at some step $c$-agreement holds on some component $c$ then $c$-agreement will keep holding for the whole protocol run.
But also, if $c$-agreement holds on some bit $b\in\{0,1\}$, then in the next Coin-Fixed-To-\emph{b} step all the honest nodes will finalize the $c$-th component, in fact the honest nodes are more than $\frac{2}{3}n$ by the assumption in \Cref{Net_Ass}.

This means that, for the unambiguous components, the honest nodes are already in $c$-agreement, therefore they will finalize such components in the first Coin-Fixed-To-0 and Coin-Fixed-To-1, i.e. by the first MBBA iteration.

The consensus evolution is much more complex for the ambiguous events. In fact agreement may not be reached in the first MBBA iteration and it might occur that a Coin-Genuinely-Flipped step is triggered.
In such case, by applying \Cref{lem3} to the case in which there is a single component $c$ on which agreement is not reached (hence $l=1$), we get that with probability at least $\frac{h}{2}$ $c$-agreement will be reached.
Once again we recall that once $c$-agreement is reached it will keep holding.

Let $\{c_1,\dots,c_l\}$ be the components associated to the ambiguous events, we have that at every Coin-Genuinely-Flipped step $c_i$-agreement, for $i\in\{1,\dots,l\}$, is reached with probability at least $\frac{h}{2}$, so we can say that the number of Coin-Genuinely-Flipped steps required to end the protocol is upper bounded by the distribution $\chi_{l,\frac{h}{2}}$ described in \Cref{game}.
Since once $c$-agreement is reached two more steps might be required to finalize the component, we can say that the operations required to end the protocol is at most:
\begin{itemize}
    \item 3 steps of MGC;
    \item $\chi_{l,\frac{h}{2}}$ iterations of MBBA;
    \item 2 steps of the next iteration necessary to finalize the last components on which $c$-agreement is reached.
\end{itemize}

Therefore the distribution of the number of MBBA iterations required to end the DTSL Protocol run is upper bounded by $1 + \chi_{l,\frac{h}{2}}$.
\end{proof}

\begin{cor}[Number of Communication Steps]
The distribution of the number of communication steps required to end the MBA protocol run is upper bounded by the random variable $ 5 + 3\chi_{l,\frac{h}{2}}$.
\end{cor}
\section{Conclusions}\label{conclusion}
We presented the MBA protocol, a Byzantine agreement protocol for synchronous and complete networks which allows the nodes to reach consensus on a vector of arbitrary values, working in parallel on each component.

The protocol we have designed is based on the extension to the multidimensional case of the protocols presented by Micali in \cite{feldman1997optimal,micali2016byzantine}, and we have presented in the analysis a probabilistic upper bound to the number of steps to be executed before the protocol halts.

We believe that the MBA Protocol would find many applications in decentralized environments, in particular in contexts in which it is required coordination between various entities that simultaneously modify the state of a decentralized system.
In particular its leaderless approach allows to widen the agreement by taking account of multiple points of view, resulting in a more democratic approach that is valued in permissionless distributed settings and also thwarts attacks from malicious leaders, which other approaches can only mitigate.

The parallel approach that tackles all components at once enhances efficiency in comparison to multiple executions of the monodimensional protocol we generalized.
In our description the identification of the messages' senders is implicit, so the advantage of the multidimensional protocol is clear only in the Coin-Genuinely-Flipped step, where just the one signature is enough to derive all the coin flips.
However, in practical settings messages are authenticated through digital signatures, and with our protocol the parallel messages are neatly organized into one, which then requires just one communication session and digital signature, further enhancing the efficiency.

\subsection{Future Works}
The network assumptions we have used allow to describe the protocol and prove its properties in a clean and intuitive way, however they are quite unrealistic in practical applications.
So the MBA protocol should be extended for usage in asynchronous and incomplete networks, that model more closely real-life communication channels.
This would obviously allow to apply the protocol to a variety of practical problems, such as blockchain platforms implementing sharding.
In fact, the MBA Protocol, if designed for asynchronous networks, would allow the nodes working on different shards to synchronize their operations creating an extremely regulated environment, which gives the right conditions for a practical reconciliation of the transactions recorded on the shards.

Another research direction could focus on extending the protocol by introducing some \emph{termination steps}, which allow to conclude the protocol execution in a predetermined number of steps if it does not halt by a certain limit.
In fact, many concrete applications benefit from an upper bound on the protocol execution time that this extension would give.
However, it is quite tricky to reconcile an execution bound with the goal to preserve as much meaningful agreement as possible: the trivial solution is to collapse the still-ambiguous components to $\bot$ so that consensus is reached in a bounded number of steps.
More advanced termination steps would be preferable, however non-trivial solutions may cause a variety of issues, especially in the setting of asynchronous incomplete networks where malicious players have a widened array of attacks at their disposal.

\paragraph{Acknowledgments}
The core of this work is contained in the first author's MSC thesis.
Part of the results presented here have been carried on within the EU-ESF activities, call PON Ricerca e Innovazione 2014-2020, project Distributed Ledgers for Secure Open Communities.
The second and third authors are members of the INdAM Research group GNSAGA.
We would like to thank the Quadrans Foundation for their support.
\\

\renewcommand{\bibname}{References} 
\bibliographystyle{plain}
\bibliography{Bibliography}
\clearpage
\end{document}